\documentclass[conference]{IEEEtran}

\usepackage[square,sort,compress,comma,numbers]{natbib}

\usepackage{xspace}
\usepackage{filecontents}

\usepackage{float}
\usepackage{pgf, tikz}

\usepackage{graphicx}

\usepackage[cmex10]{amsmath}

\usepackage{bm}
\usepackage{soul,color} %
\usepackage{amssymb} 

\usepackage{array}

\usepackage{subfigure}
\usepackage{caption2} 

\usepackage{amsfonts}
\usepackage{amsthm}   %

\usepackage{enumerate} %

\usepackage{extarrows}

\usepackage{tkz-fct}

\newtheorem{assumption}{Assumption} 
\newtheorem{theorem}{Theorem} 
\newtheorem{lemma}{Lemma}
\newtheorem{proposition}{Proposition}
\newtheorem{corollary}{Corollary}

\newtheorem{remark}{Remark}

\newtheorem{definition}{Definition}

\usetikzlibrary{patterns,snakes}

\usepackage{enumerate}

\usepackage{extarrows}
\usepackage{balance}

\newcommand{\myexpect}[2]{\mathsf{E}_{#1}\left[ #2 \right]}

\newcommand{\myprobability}[1]{\mathrm{Pr}\!\left[ #1 \right]}

\newcommand{\snr}{\mathsf{SNR}}

\newcommand{\tr}[1]{\left(#1\right)}
\newcommand{\myA}{\mathcal{A}}

\newcommand{\mynorm}[1]{\left\lvert #1 \right\rvert^2}

\newcommand{\myQQ}{\mathcal{Q}}

\newcommand{\patt}{\text{Patt}}

\begin{document}
\title{\Huge On the Latency, Rate and Reliability Tradeoff in Wireless Networked Control Systems for IIoT}
\author{\IEEEauthorblockN{Wanchun Liu, \emph{Member, IEEE}, Girish Nair,  \emph{Fellow, IEEE}, Yonghui Li, \emph{Fellow, IEEE},  Dragan Nesic, \emph{Fellow, IEEE},\\  Branka Vucetic, \emph{Fellow, IEEE},  and H. Vincent Poor, \emph{Fellow, IEEE}}
\thanks{\setlength{\baselineskip}{13pt} \noindent W. Liu, Y. Li and B. Vucetic are with School of Electrical and Information Engineering, the University of Sydney, Australia.
(emails: \{wanchun.liu,\ yonghui.li,\ branka.vucetic\}@sydney.edu.au).
G. Nair and D. Nesic are with  Department of Electrical and Electronic Engineering, The University of Melbourne, Australia.
(emails: \{gnair,\ dnesic\}@unimelb.edu.au).
}
}

\maketitle
\begin{abstract}
	\let\thefootnote\relax\footnote{
		The work of
		W. Liu was supported by the Faculty of Engineering Early Career
		Researcher Development Scheme 2020, University of Sydney. The work
		of Y. Li was supported by ARC under Grant DP190101988. The
		work of B. Vucetic was supported by the Australian Research Council's
		Australian Laureate Fellowships Scheme under Project FL160100032.
W. Liu, Y. Li and B. Vucetic are with School of Electrical and Information Engineering, the University of Sydney, Australia.
(emails: \{wanchun.liu,\ yonghui.li,\ branka.vucetic\}@sydney.edu.au).
G. Nair and D. Nesic are with  Department of Electrical and Electronic Engineering, The University of Melbourne, Australia.
(emails: \{gnair,\ dnesic\}@unimelb.edu.au).
H. Vincent Poor is with Department of Electrical Engineering, Princeton University. (email: poor@princeton.edu).
}
Wireless networked control systems (WNCSs) provide a key enabling technique for Industry Internet of Things (IIoT). 
However, in the literature of WNCSs, most of the research focuses on the control perspective, and has considered oversimplified models of wireless communications which do not capture the key parameters of a practical wireless communication system, such as latency, data rate and reliability.
In this paper, we focus on a WNCS, where a controller transmits quantized and encoded control codewords to a remote actuator through a wireless channel, and adopt a detailed model of the wireless communication system, which jointly considers the inter-related communication parameters.
{We derive the stability region of the WNCS. If and only if the tuple of the communication parameters lies in the region, the average cost function, i.e., a performance metric of the WNCS, is bounded.}
We further obtain a necessary and sufficient condition under which the stability region is $n$-bounded, where $n$ is the control codeword blocklength.
We also analyze the average cost function of the WNCS.  Such analysis is non-trivial because the finite-bit control-signal quantizer introduces a non-linear and discontinuous quantization function which makes the performance analysis very difficult.
We derive tight upper and lower bounds on the average cost function in terms of latency, data rate and reliability.   
Our analytical results provide important insights into the design of the optimal parameters to minimize the average cost within the stability region.
\end{abstract}


\section{Introduction}
\noindent
Wireless networked control systems (WNCSs) have many applications in industrial and building automation, intelligent transportation systems and smart grid. This technology is driven by 
recent advances in wireless communications, networking, sensing, computing,
and control, as well as Industrial Internet of Things (IoT) applications~\cite{ParkSurvey,Turkey,Schulz,Splitting1,Splitting2,Jiang1,Jiang2,Jiang3}.
{In general, a WNCS is a spatially distributed control system consisting of a dynamic plant, a set of sensors that measure and report the plant state, a remote controller that collects the sensors' measurement and generates control signals, and a set of actuators that control the plant based on the received control signals.}
In contrast to the conventional NCSs using wired communications that support high-rate, real-time and reliable data transmission,
a WNCS can only transmit sensors' measurement and the control signals through unreliable wireless channels with relatively low signal-to-noise ratios (SNR), and thus have limited communication performance in terms of latency, data rate and reliability, which then impose constraints on the performance of the WNCS.

In the literature on WNCSs, various models of wireless communication
systems have been investigated.
A data-rate theorem that states the minimum data rate needed to stabilize a linear plant assuming a noiseless communication system with a limited data rate was derived in~\cite{Nair2004}.
The stability condition in terms of the SNR was investigated in~\cite{LiuGC,LiuJIoT} assuming coding-free (analog) communications.
The stability conditions in terms of the packet dropout probability have been derived in~\cite{TCP} assuming an infinite data rate (i.e., zero quantization error) and independent and identically distributed (i.i.d.) packet dropouts.
Based on this assumption, an optimal retransmission scheduling problem and an optimal sensing-control transmission scheduling problem of WNCSs were investigated in~\cite{KangTWC} and~\cite{KangJIoT}. In addition, multi-WNCS scheduling problems were investigated in~\cite{AlexLeong,AlexLeong2,Burak,Eisen,Eisen2}.
In~\cite{Nair09}, a linear scalar WNCS was considered, taking into account \emph{independent} communication parameters.

All the models of wireless communication systems considered in \cite{Nair2004,LiuGC,LiuJIoT,TCP,KangTWC,KangICC,KangJIoT} are valid under restrictive (and often unrealistic) assumptions and do not simultaneously consider the effect of latency, data rate and reliability, i.e., the three performance parameters of a high fidelity model for practical wireless communication systems.
%
In \cite{Nair09}, the effect of the interrelation between the communication parameters on the WNCS was ignored.
However, it is well-known that in a communication system, these parameters are inter-related, i.e., one cannot freely change one parameter without affecting others~\cite{tradeoffs}.
This naturally leads to {joint parameter design problems} in WNCSs, which have not been considered in the open literature.
Furthermore, the control performance in terms of average cost function of a WNCS with quantization and channel-coding methods remains largely unknown. 
Thus, the optimal design of the communication parameters to minimize the average cost of a control system in addition to stabilizing it, is still an open problem even for a scalar system.

In this paper, we focus on a scalar WNCS~\cite{Nair09,LiuGC,LiuJIoT}, where the controller sends control signals through an additive white Gaussian noise (AWGN) channel to the actuator which controls the plant system~\cite{TCP,KangTWC,KangICC,KangJIoT}. 
Our target is to obtain the stability region of the plant in terms of the inter-related communication parameters and analyze the average cost function of the plant within the stability region.

{The main contributions are summarized as follows:
\begin{itemize}
\item We investigate a WNCS with an ideal control and quantization method, where the controller and the actuator can utilize all the historical control and communication state information for control-signal quantization and de-quantization.
We derive the \emph{stability region} of the plant in terms of the channel-coding blocklength $n$ and the data rate $R$ based on the data-rate theorem~\cite{Nair09} and the finite blocklength information theory~\cite{Polyanskiy}.
The average cost function of the WNCS is bounded if and only if the tuple of communication parameters lies within the stability region.
We further derive the necessary and sufficient condition under which the stability region is $n$-bounded.
\emph{The result shows an important counter-intuitive finding: the plant can be stabilized with an arbitrarily large blocklength (i.e., a large latency) as long as the SNR of the wireless channel is greater than a certain value determined by the control-system parameters.}

\item We investigate a WNCS with a practical control and quantization method where the controller and the actuator can only utilize the current control and communication state information for quantization and de-quantization. Specifically, we adopt a zooming quantizer that can adaptively change its quantization range.
We derive the necessary and sufficient condition of the existence of a zooming quantizer in terms of the channel-coding blocklength $n$ and the data rate $R$. The pre-quantized control signal is always within the quantization range if and only if the condition is satisfied.

\item For the practical WNCS, we derived the stability region and obtain the necessary and sufficient condition under which the stability region is $n$-bounded.
We also analyze the average cost function of the WNCS, which is non-trivial. 
This is because the finite-bit control-signal quantizer introduces a non-linear and discontinuous quantization function which makes the performance analysis very difficult.
We derive closed-form upper and lower bounds of the average cost function in terms of the communication parameters. Also, we show that the gap between the bounds is negligible when $n$ is large, and the upper bound is tight even when $n$ is small.  
Our performance analysis helps to explain the above mentioned counter-intuitive finding: \emph{even though a plant can be stabilized with an arbitrarily large $n$, the average cost of the system is arbitrarily large and is not desirable.}

\end{itemize}}

The remainder of the paper is organized as follows: 
Sec.~II describes the WNCS, where the controller sends quantized and coded control signal to the actuator to control the dynamic plant.
Sec. III establishes the fundamental stability condition of the WNCS in terms of the communication parameters under the assumption that the quantizer can utilize all the historical control and communication state information and ideally require infinite memory. 
In Sec. IV-V, we investigate a WNCS with practical control and quantization methods with limited memory, and derive the stability condition and analyze the average cost function of the WNCS.
Finally, Sec. VI concludes this work.

Notation: $\mathbb{Z}_0$ and $\mathbb{Z}_+$ denote the set of non-negative and positive integers, respectively. $\mathbb{R}$ and $\mathbb{C}$ are the sets of real and complex numbers, respectively.
$X_t$ denotes a random process, and $x_t$ denotes its realization.
$\mathsf{Pr}[\cdot]$ and $\mathsf{E}[\cdot]$ denote probability and expectation operators.
$\limsup$ is the limit superior operator.

\begin{figure}[t]
	\centering
	\begin{tikzpicture}[scale=0.83]
	\draw [rounded corners=5pt] (-2.5,4.5) rectangle (1.5,2.5);
	\draw [rounded corners=5pt] (-3.5,1) node (v1) {} rectangle (2.5,0);
	\node [align=center] at (-0.5,3.5) {Plant\\$x_{t+1} \!=\!a x_{t} \!+ b u_{t} \!+ w_{t}$};
	\node at (-0.5,0.5) {Controller};
	\draw  (-4.5,4) rectangle (-2.5,3);
	\draw  (1.5,4) rectangle (3.25,3);
	\node at (-3.5,3.5) {Actuator};
	\node at (2.35,3.5) {Sensor};
	
	\node at (-4.15,-0.05) {$\tilde{\mu}_t$};
	\node at (3.65,3.75) {$x_t$};
	\draw [thick,-latex,rounded corners=5pt] (3.25,3.5) -- (4.5,3.5) -- (4.5,0.5) -- (2.5,0.5);
	\node at (-3.5,4.25) {$u_t$};
	\draw [-latex,thick,double distance=2pt,rounded corners=5pt](-3.5,0.3) -- (-6.2,0.3) -- (-6.2,3.7) -- (-4.5,3.7);
	
	\draw [] plot[smooth, tension=.7] coordinates {(4,1.8) (4.15,1.9) (4.35,1.7) (4.5,1.85)};

	\draw [] plot[smooth, tension=.7] coordinates {(-6.15,2.15) (-5.85,2.35) (-5.6,2.15) (-5.3,2.2)};
	\node [right] at (-5.3,2.15) {Backward Channel};

	\draw [] plot[smooth, tension=.7] coordinates {(-5.75,1.35) (-5.5,1.55) (-5.25,1.35) (-4.95,1.4)};
	\node [right] at (-4.95,1.35) {Feedback Channel};

	\node at (2.5,1.8) {Forward Channel};
	\draw [dashed,thick,-latex,rounded corners=5pt](-4.5,3.3) -- (-5.8,3.3) -- (-5.8,0.65) -- (-3.5,0.65);
	\node at (-4.15,0.85) {$s_t$};
	\end{tikzpicture}
	\caption{An illustration of the WNCS, where $x_t$, $\tilde{\mu}_t$, $u_t$, $s_t$ and $w_t$ are the plant state, the control signal, the control action, the feedback signal and the disturbance at time $t$, respectively. $a$ and $b$ are the parameters of the plant system.}
	\label{fig:sys}
\end{figure}
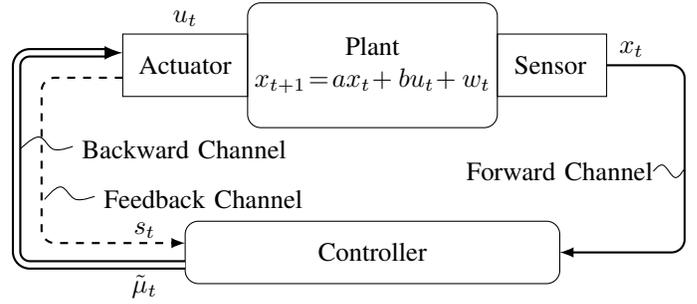

\section{System Model}
\noindent
We focus on a linear time-invariant (LTI) discrete-time scalar WNCS~\cite{Nair09,LiuGC,LiuJIoT} consisting of a dynamic unstable plant, an actuator, a wireless sensor, a remote controller, a \emph{forward} (i.e., sensor-to-controller) and a \emph{backward} (i.e., controller-to-actuator) digital channels, as shown in Fig.~\ref{fig:sys}. 
In general, based on the plant state information sent by the sensor through the forward channel, the controller generates an analog control signal, which is then quantized and coded by the quantizer and the channel-coding module, respectively. Then, the coded control signal (i.e., a sequence of control-information-carrying symbols) is sent to the actuator through the backward channel, and the actuator controls the plant and sends a feedback signal to the controller through the feedback channel.

We assume that the sensor and the controller are collocated, while the actuator and the controller are dislocated, and thus only focus on the \emph{backward} communication channel.
In other words, it is assumed that the state information of the plant is perfectly known by the controller without any delay~\cite{DanielBernoulli,DanielTradeoff,LiuGC,LiuJIoT,KangWCL}. 
However, in order to send the control-information bearing signal to the actuator through a band-limited wireless channel with the constraints of latency, data rate and reliability, the controller needs to perform adaptive quantization (i.e., source coding) and channel coding.
For low-mobility industrial control applications,
the communication channel is static, and thus the backward channel can be modeled as an AWGN channel~\cite{TCP,KangTWC,KangICC,KangJIoT}.\footnote{The fading channel scenario will be investigated in our future work.}
For simplicity, we assume that the discrete time step (i.e., the sampling period) of the dynamic system is equal to the symbol duration of the backward information transmission~\cite{Nair09}.

The dynamic plant system, communications through AWGN backward channel, quantization and control policy are described as follows.

\subsection{Dynamic Plant System}
The sensor periodically measures and sends the state of the dynamic system (i.e., the plant) to the controller, which utilizes such information to generate and send control signals to the actuator to stabilize the plant.
We consider a scalar discrete-time plant model described by~\cite{Nair09,LiuGC,LiuJIoT}
\begin{equation}\label{system}
x_{t+1} = a x_t + b u_{t} + w_t,\ \forall t\in \mathbb{Z}_0,
\end{equation}
where $x_t\in \mathbb{R}$ is the plant state at $t$, $u_{t}\in \mathbb{R}$ is the control action imposed by the actuator. 
We assume that the initial state, $x_0$, has {bounded support}, i.e., $x_0\in \left(-x_{0,\max},x_{0,\max}\right)$.
The i.i.d. process $w_t$ with zero mean and variance $\sigma^2_W$ is the plant disturbance, and $w_t \in \left(-w_{\max},w_{\max}\right)$.
The real numbers $a$ and $b$ are the parameters of the system. Specifically, we assume that $\vert a \vert >1$ implies that the plant is open-loop unstable, i.e., the plant state grows up unbounded exponentially fast without the control input $u_t$.

\subsection{Communications Through AWGN Backward Channel} \label{coding}
Consider a communication system with a fixed transmission rate and finite number of alphabet inputs. 
From time $t=0$, the transmitter (i.e., the controller) starts sending channel codewords every $n$ slots to the receiver (i.e., the actuator), where $n$ is the blocklength of a codeword. 
We denote a codebook with $2^{nR}$ codewords by $\mathcal{C} \triangleq \left\lbrace c_1,c_2,...,c_{2^{nR}} \right\rbrace$, where $R$ is the data rate (which is also known as the coding rate) and $c_i \in \mathbb{C}^n$, $\forall i$. The codebook is known by both the transmitter and the receiver.
The average decoding-error probability of this channel coding scheme through an AWGN channel is denoted as $\epsilon$.
Thus, the codebook, $\mathcal{C}$, can be denoted as an $(n,R,\epsilon)$ code from an information-theoretic perspective~\cite{Polyanskiy}.

Let ${\tilde{\mathcal{C}}}(\cdot)$ denote the channel coding function that maps the index into the set $\mathcal{C}$.
Let $\mu^{t+n}_{t+1} \in \mathcal{C}$ denote the output codeword of the channel encoder ${\tilde{\mathcal{C}}}(\cdot)$ that is transmitted in $n$ time slots from $t+1$, where the notation $z_n^k$ represents the sequence of $z_i$ from $n$ to $k$.

We make the following assumption:
\begin{assumption}
	\normalfont
The receiver knows whether the previous and current codewords have been correctly decoded or not, and the receiver discards the codeword once a decoding error occurs~\cite{TCP}.
\end{assumption}
Note that the error detection can be realized by adding a few redundant symbols for cyclic redundancy check (CRC). As we only focus on the information-carrying symbols in this work, the length of the CRC symbols is ignored for tractability\footnote{Although the CRC can never be 100\% correct in a wireless environment with unbounded noise, in the high SNR scenario the detection success rate can be made arbitrarily high and can be ignored.}.
The receiver sends a one-bit feedback to the transmitter through a perfect feedback channel, which indicates the correct/incorrect decoding.
Let $s_t$ indicate the correct/incorrect decoding of a codeword at time slot $t$, as illustrated in Fig.~\ref{fig:sys}. If a codeword is correctly decoded at $t$, $s_t = 1$; otherwise, $s_t = 0$.
Note that such feedback information is only utilized for quantization and control-signal generation rather than channel coding. As proved in \cite{Polyfeedback}, in the finite blocklength scenario, the achievable data rate with a feedback can only improve at most $2-3$ bits per block when $n>10$ compared with the non-feedback case.

{
Following the finite-blocklength information theory~\cite{Polyanskiy},
the channel capacity $C$ and the channel dispersion $\nu$ are defined as
\begin{align}
\label{C}
C &= \log_2\left(1+\snr\right),\\
\label{v}
\nu &=\snr \frac{2+\snr}{\left(1+\snr\right)^2} \left(\log_2 e\right)^2,
\end{align}
where $\snr$ is the SNR at the actuator and $e\approx 2.718$ is the Euler's number.}
Given the reliability (i.e., $\epsilon$) and latency (i.e., $n$) requirement, 
the maximum achievable data rate is approximated for $n \geq 100$ as~\cite{Polyanskiy}
\begin{equation}\label{Yury_R}
R \approx C- \sqrt{\frac{\nu}{n}}Q^{-1}\left(\epsilon\right),
\end{equation} 
where $Q(x) = \frac{1}{2 \pi} \int_{x}^{\infty} \exp\left(-\frac{u^2}{2}\right) \mathrm{d}u$ and $Q^{-1}(\cdot)$ is the inverse function of $Q(\cdot)$.
Thus, \eqref{Yury_R} demonstrates a design tradeoff between the latency, the reliability and the data rate (i.e., $R$);
given the data rate and latency requirement, the minimum achievable decoding-error probability is approximated as
\begin{equation} \label{Yury_e}
\epsilon \approx Q\left(\sqrt{\frac{n}{\nu}}\left(C-R\right)\right).
\end{equation}

\subsection{Quantization and Control} \label{sec:quantization and control}
We consider a general control-information quantizer
$\mathbb{\tilde{Q}}_t\left(x_0^t, \mu_1^{t},s_1^{t}\right) \in \{1,2,\cdots,2^{nR}\}, \forall t= k n, k \in \mathbb{Z}_0$, which generates and quantizes the control signal based on all the previous and current plant states, control codewords, and detection feedback~\cite{Nair2004,Nair09}.
Then, the index of the quantized control signal is passed to the channel encoder, and the output of the channel encoder from time $t+1$ to $t+n$ is represented as
\begin{equation} \label{coder}
\begin{aligned}
\mu^{t+n}_{t+1} &= 
\mathcal{E}_t\left(x_0^t, \mu_1^{t},s_1^{t}\right)\\
&\triangleq
{\mathcal{\tilde{C}}}\left(\mathbb{\tilde{Q}}_t\left(x_0^t, \mu_1^{t},s_1^{t}\right)\right)
, \forall t= k n, k \in \mathbb{Z}_0
\end{aligned}
\end{equation}
where $\mathcal{E}_t(\cdot)$ can be treated as the encoder at the controller for quantization and channel coding.

The decoder at the actuator is given as
\begin{equation} \label{decoder}
{u}_t = \mathcal{D}_t\left({\mu}_1^t \odot s_1^{t}\right), \forall t= k n, k \in \mathbb{Z}_+
\end{equation}
where $u_t \in \mathbb{R}$ is the actuator's decoded and de-quantized control signal (i.e., the control action), ${\mu}_1^t \odot s_1^{t}$ only keeps the correctly detected parts in ${\mu}_1^t$,
and the decoding function $\mathcal{D}_t(\cdot)$ depends on all the previous and current control codewords that have been detected correctly.

\begin{remark}
For the WNCS in \eqref{system} with the channel coding scheme described in Sec.~\ref{coding}, and the quantization and control policy in \eqref{coder} and \eqref{decoder}, the data rate $R$ and codeword blocklength $n$ jointly determine the quantization resolution of the control signal, the blocklength $n$ determines the control latency (and frequency), and the decoding-error probability $\epsilon$ determines the transmission reliability of the control signal. Therefore, the control performance of the WNCS depends on $(n,R,\epsilon)$.
\end{remark}

Note that the general encoder and decoder in \eqref{coder} and \eqref{decoder} utilizing all the historical information ideally require infinite memory for control-information encoding and decoding.
We will consider the stability condition of the WNCS without encoder/decoder memory constraint, and investigate a practical WNCS with memory constraint in the sequel.

\section{Stability Region of the WNCS without Memory Constraint}\label{sec:conter}
\noindent
In this section, we focus on the necessary and sufficient conditions for stabilizing the plant system \eqref{system} in the mean-square sense, which is defined as~\cite{Nair2004,Nair09,KangGC,KangJIoT,LiuGC,LiuJIoT,KangWCL} 
\begin{equation}\label{first_condition}
\limsup_{t\rightarrow \infty} \myexpect{\!}{\mynorm{x_t}} < \infty,
\end{equation}
for any initial state $x_0\in \left(-x_{0,\max},x_{0,\max}\right)$.
{The long-term average cost of the stabilized plant state in the mean-square sense is defined~as~\cite{Nair09,Nair2004} 
\begin{equation} \label{cost}
\myA \triangleq \limsup_{t\rightarrow \infty} \myexpect{\!}{\mynorm{x_t}}.
\end{equation}

Note that the average cost of the plant indicates the mean-square error (MSE) performance of the plant state. 
For example, we consider a speed control problem of an automated guided vehicle (AGV) in IIoT applications.
Let $x'(t)$ and $\tilde{v}$ denote the real speed at time $t$ and the default speed, respectively.
Thus, $x_t \triangleq x'_t-\tilde{v}$ denotes the error of speed at time $t$.
Then, the MSE of the speed at $t$ is $\myexpect{\!}{\mynorm{x_t}}$. 
When considering the long-term performance of the speed-control process, $\limsup_{t\rightarrow \infty} \myexpect{\!}{\mynorm{x_t}}$ can be treated as an upper bound of the long-term average MSE~\cite{Nair2004,Nair09}.}

Based on the same approach adopted in \cite{Nair09}, we can obtain the following stability condition:
\begin{lemma}\label{lemma:1}
	\normalfont
Given the channel coding scheme in Sec.~\ref{coding}, and the quantization and control policy in \eqref{coder} and \eqref{decoder}, the necessary and sufficient condition for stabilizing the plant \eqref{system} in the mean-square sense \eqref{first_condition} is:
\begin{equation}
\epsilon a^{2 n} + (1-\epsilon) \frac{a^{2 n}}{2^{2nR}} < 1.
\end{equation}
\end{lemma}
\begin{proof}
	See Appendix A
\end{proof}

As described in Sec.~\ref{coding}, $n$ and $R$ are the key design parameters of the codebook for channel coding, while $\epsilon$ can be treated as a performance measure of the codebook. Thus, we consider a family of optimal channel codes $(n, R)$ that achieves the minimum error probability in \eqref{Yury_e}. Thus, we have the following definition.
\begin{definition}[Stability Region] \label{def:stability}
	\normalfont
	Given the channel coding scheme in Sec.~\ref{coding}, and the quantization and control policy in \eqref{coder} and \eqref{decoder},
	the stability region of the plant \eqref{system}, is the set of the communication parameters $(n,R)$ of the optimal channel codes
	such that the plant is stabilized in the mean-square sense~\eqref{first_condition}.
\end{definition}

Based on Definition~\ref{def:stability}, by taking \eqref{Yury_e} into Lemma~\ref{lemma:1}, we have the stability region of the plant in terms of the design parameters $(n, R)$ of the communication system.

\begin{corollary}\label{theo:1}
	\normalfont	
	Assuming sufficiently large blocklength, the stability region of the plant \eqref{system} is approximated as
	\begin{equation}
\begin{aligned}
	\mathcal{S}:=
	&\left\lbrace  (n, R): 
		Q\left(\sqrt{\frac{n}{\nu}}\left(C-R\right)\right) a^{2 n} \right.\\		
	&\left.	+\! \left(\!1 - \! Q\left(\!\sqrt{\frac{n}{\nu}}\left(C-R\right)\!\right)\!\right)\! \frac{a^{2 n}}{2^{2nR}} < 1, \forall n \in \mathbb{Z}_+, R>0
	\!\right\rbrace.
\end{aligned}
	\end{equation}
\end{corollary}
\begin{remark}
If a pair $(n,R) \in \mathcal{S}$, there exist a channel coding scheme $\mathcal{C}$ and a pair of encoder-decoder $\left(\mathcal{E}_t,\mathcal{D}_t\right)$ that stabilize the plant. Otherwise, the plant cannot be stabilized.
\end{remark}

From Corollary~\ref{theo:1}, we have the the following property of the stability region.
\begin{proposition} \label{prop:bound_nair}
	\normalfont	
	Assume the stability region is not empty based on Corollary~\ref{theo:1}. If $\log_2 a \geq C-2 \sqrt{\nu\ln a}$, the two-dimensional stability region is $n$-bounded, otherwise, the region is $n$-unbounded when $R \in \left(\log_2 a, C- 2 \sqrt{\nu \ln a} \right)$.
\end{proposition}
\begin{proof}
See Appendix B.
\end{proof}

By using Corollary~\ref{theo:1}, the stability regions with different $\snr$ conditions are plotted in Fig.~\ref{fig:region-nair}. From Proposition~\ref{prop:bound_nair}, the stability region is $n$-bounded when $\snr \leq 10.3$~dB, which has been verified in Fig.~\ref{fig:region-nair} that the region with $\snr =10$~dB is $n$-bounded, while the regions with $\snr = 11$~dB and $20$~dB are $n$-unbounded within the $R$-interval of $\left(\log_2 a, C- 2 \sqrt{\nu \ln a} \right)$.

\begin{figure}[t]
	\renewcommand{\captionfont}{\small} \renewcommand{\captionlabelfont}{\small}
	\centering
	\includegraphics[scale=0.67]{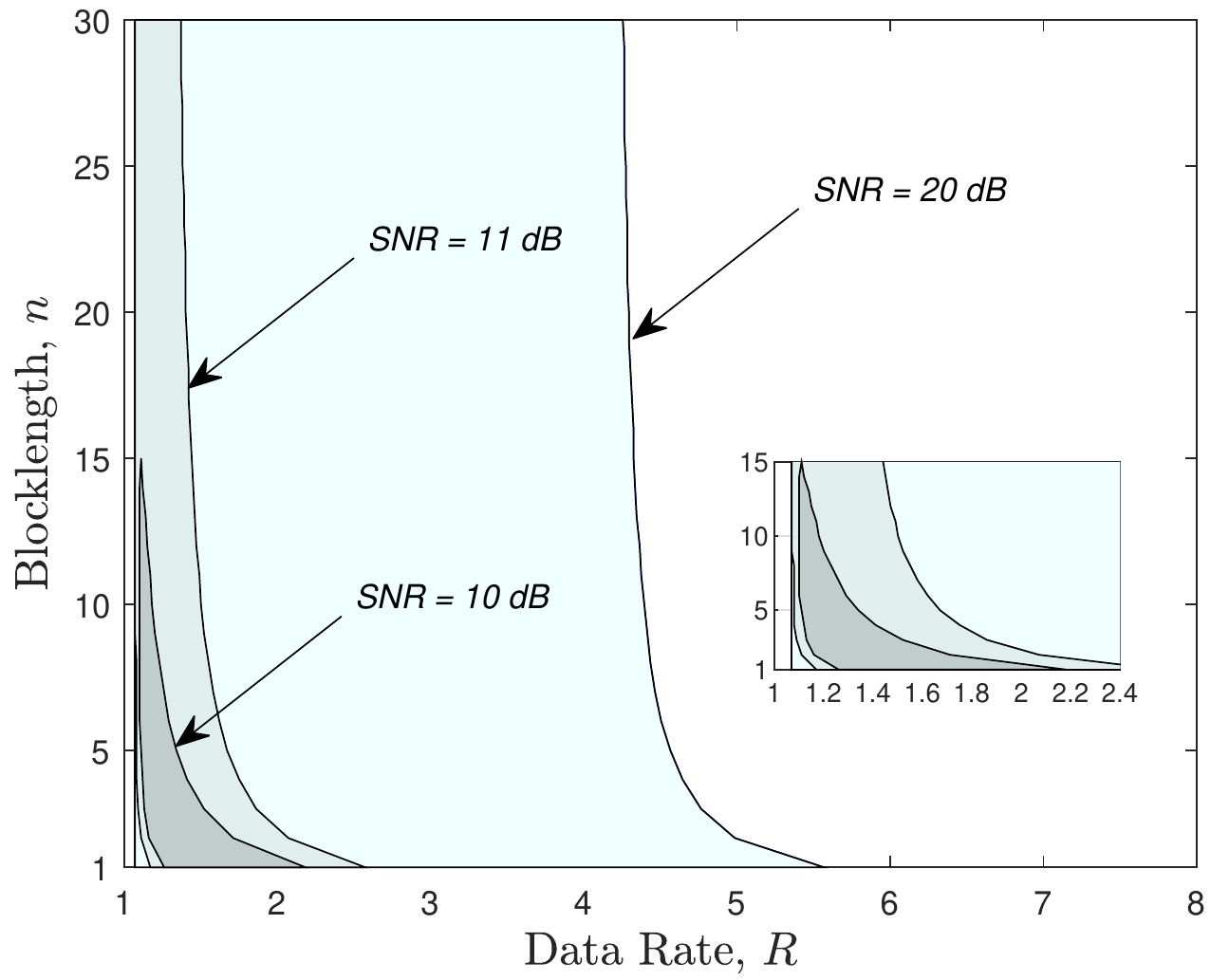}
	\caption{The stability region of a plant in terms of the data rate and the blocklength, where $a = 2.1$ and $b=1$. The bounded dark colored region and the unbounded moderate and light colored regions are plotted at $\snr = 10$, $11$ and $20$ dB, respectively.}
	\label{fig:region-nair}
\end{figure}

\begin{remark}
From Proposition~\ref{prop:bound_nair}, an $n$-unbounded stability region suggests that the plant can be stabilized even with an arbitrarily large latency caused by channel coding.	
	
Based on the definitions in \eqref{C} and \eqref{v}, it can be proved that $C-2 \sqrt{\nu\ln a}$ monotonically increases from zero to infinity as $\snr$ increases.
Therefore, in Proposition~\ref{prop:bound_nair}, the stability region is $n$-unbounded if and only if the SNR is greater than the one yielding the equality in Proposition~\ref{prop:bound_nair}, i.e., $\log_2 a = C-2 \sqrt{\nu\ln a}$.
\end{remark}

Intuitively, since a long blocklength introduces a large delay for control, it seems that the plant cannot be stabilized. However, the stability region result indicates that the plant can be stabilized with an arbitrarily large blocklength as long as the SNR is larger than a certain value determined by the parameter of the plant.
In the rest part of the paper, we can explain this counter-intuitive result based on the performance analysis of the WNCS.

\section{Modeling of a Practical WNCS}\label{sec:new}
\noindent
In the previous section, although we have obtained the stability region of the plant with ideal quantization/de-quantization method that utilizes the historical information and requires infinite memory, it is still unclear how to optimally design the parameters within the stability region such that the average cost of the WNCS is minimized.
The main difficulty is due to the fact that the average cost function does not have an explicit expression in terms of $n$ and $R$, as a finite-bit quantizer can always introduce a non-linear and discontinuous quantization function which makes the performance analysis of the dynamic process $x_t$ in the mean-square sense not only difficult but impossible.

In this section, we focus on the WNCS with practical control and quantization methods that only depend on the current status of the WNCS and require small memory for encoding and decoding at the controller and the actuator. Such a practical WNCS will also enable tractable control-performance analysis of the average cost function.

\subsection{Control Policy} \label{sec:control}
Let $\tilde{\mu}_t$ denote the control signal generated at time $t$ and expected to be received by the actuator at time $t+n$ due to the transmission latency.
We adopt the commonly used linear–quadratic–Gaussian (LQG) method for generating the control signal $\tilde{\mu}_t$~\cite{TCP}, i.e.,
\begin{equation} \label{control_sig}
\tilde{\mu}_t = \kappa \hat{x}_{t+n},
\end{equation}
where $\kappa = - a/b$ for the scalar system~\eqref{system} and $\hat{x}_{t+n}$ is the estimated system state at time $(t+n)$.

We assume that the actuator adopts a zero-hold strategy~\cite{TCP,KangTWC,KangICC,KangJIoT}: the actuator controls the plant only in the time slot in which a received codeword has been decoded, and does not control the plant in the other time slots by letting the control signal to be equal to zero. {This assumption is also motivated by the data integrity requirement in the industrial control applications. In other words, the actuator cannot estimate the plant state nor generate control action by itself and can only control the plant when it receives the control signal.}

Since there is no control action between $t$ and $t+n$, the optimal state estimation can be obtained from \eqref{system} as~\cite{DanielTradeoff,KangTWC,KangICC}
\begin{equation} \label{estimat}
\begin{aligned}
\hat{x}_{t+n} &= a^{n-1} \hat{x}_{t+1}\\
&= a^{n-1} \times 
\begin{cases}
a x_t & t\!= k n, k\! \in \mathbb{Z}_0, s_t\!=0,\\
(a x_t+ b u_t) & t\!= k n, k\! \in \mathbb{Z}_0, s_t\!=1.
\end{cases}
\end{aligned}
\end{equation}

Taking \eqref{estimat} into \eqref{control_sig}, the control signal at the controller before quantization can be written~as
\begin{equation} \label{command}
\tilde{\mu}_t  = 
\begin{cases}
-\frac{a^{n+1}}{b} {x}_t
, & t\!= k n, k\! \in \mathbb{Z}_0, s_t\!=0\\
-\frac{a^{n}}{b} \left(a x_t + b {u}_{t}\right)
, & t\!= k n, k \!\in \mathbb{Z}_0, s_t\!=1\\
 0, & \text{else}.
\end{cases}
\end{equation}
Then, the generated control codeword at time $kn$, which will be transmitted from $(t+1)$ to $(t+n)$, is given as
\begin{equation}
 \mu^{t+n}_{t+1} = \tilde{\mathcal{C}}\left( \hat{\mathbb{Q}} \left(\tilde{\mu}_t\right)\right) \triangleq 
 c_{\hat{\mathbb{Q}}\left( \tilde{\mu}_t\right) },
\end{equation}
where the function $\hat{\mathbb{Q}}(\cdot)$ outputs the index number of the quantized signal, and $\tilde{\mathcal{C}}(\cdot)$ and $c_i$ have been defined in Sec.~\ref{coding}.

Therefore, the actuator's control action can be written as
\begin{equation} \label{action}
u_t= \mathcal{D}(\mu^{t}_{t+1-n})  \triangleq 
\left\lbrace
\begin{aligned}
&\tilde{\mu}_{t-n} + v_{t}, && t= k n, k \in \mathbb{Z}_+, s_t=1,\\
& 0, && \text{else},
\end{aligned}
\right. 
\end{equation}
where $v_t$ is the quantization noise caused by the $(n R)$-bit quantizer at time $t-n$.

Fig.~\ref{fig:illu-process} illustrates the plant state controlled by the control codewords with correct and incorrect decoding.
It is clear that if a consecutive incorrect detection occurs, the absolute value of the plant state can be very large, while if a correct detection occurs, the plant state can be close to zero.
Thus, the value of the control signal in \eqref{command} has a large range. 
To handle the quantization of the wide-range dynamic signal, we need a quantizer that can adaptively change its quantization range.

\begin{figure}
	\centering
\begin{tikzpicture}[scale=0.9]

\draw [-latex] (-4,0) node[scale=0] (v1) {} -- (4,0);
\draw [-latex] (-4,-1.4) -- (-4,2.9);

\draw (-4,0.5) -- (-3.5,0.5);
\draw (-3.5,0.5) .. controls (-2.9,0.6) and (-2.6,1.35) .. (-2.5,2) node (v2) {};
\node (v3) at (-2.5,-0.3) {};
\node at (-2,-0.3) {};
\draw (-2,-0.3) node (v4) {} .. controls (-1.5,-0.3) and (-1,-0.8) .. (-1,-1.4) node (v5) {};
\draw (-2.5,2) -- (-2.5,-0.3) -- (-2,-0.3);
\draw (-1,-1.4) -- (-1,0.25) -- (-0.5,0.25);
\draw (-0.5,0.25) .. controls (0.8,0.5) and (1.5,1.25) .. (2,2.45) node (v6) {};
\draw (2,2.45) -- (2,0.35) -- (2.5,0.35);
\draw (2.5,0.35) .. controls (3,0.4) and (3.4,0.65) .. (3.5,1.35);
\node at (-3.35,2.7) {$x_t$};
\node at (3.85,-0.5) {$t$};
\draw  [thick] (-3.5,0.05) -- (-3.5,-0.1);
\draw [thick]  (-3,0.05) -- (-3,-0.1);
\draw [thick] (-2.5,-0.1) -- (-2.5,0.05);
\draw [thick] (-2,0.05) -- (-2,-0.1);
\draw [thick] (-1.5,-0.1) -- (-1.5,0.05);
\draw [thick] (-1,0.05) -- (-1,-0.1);
\draw  [thick](-0.5,0.05) -- (-0.5,-0.1);
\draw [thick] (0,0.05) -- (0,-0.1);
\draw [thick] (0.5,0.05) -- (0.5,-0.1);
\draw [thick] (1,0.05) -- (1,-0.1);
\draw [thick] (1.5,0.05) -- (1.5,-0.1);
\draw [thick] (2,0.05) -- (2,-0.1);
\draw [thick] (2.5,0.05) -- (2.5,-0.1);
\draw [thick] (3,0.05) -- (3,-0.1);
\draw [thick] (3.5,0.05) -- (3.5,-0.1);
\node [fill=blue,circle,scale=0.5] at (-3.5,0.5) {};
\node [fill=blue,circle,scale=0.5] at (-3,0.8) {};
\node [fill=blue,circle,scale=0.5] at (-2.5,2) {};
\node [fill=blue,circle,scale=0.5] at (-2,-0.3) {};
\node [fill=blue,circle,scale=0.5] at (-1.5,-0.45) {};
\node [fill=blue,circle,scale=0.5] at (-1,-1.4) {};
\node [fill=blue,circle,scale=0.5] at (-0.5,0.25) {};
\node [fill=blue,circle,scale=0.5] at (0,0.4) {};
\node [fill=blue,circle,scale=0.5] at (0.5,0.6) {};
\node [fill=blue,circle,scale=0.5] at (1,0.95) {};
\node [fill=blue,circle,scale=0.5] at (1.5,1.5) {};
\node [fill=blue,circle,scale=0.5] at (2,2.45) {};
\node [fill=blue,circle,scale=0.5] at (2.5,0.35) {};
\node [fill=blue,circle,scale=0.5] at (3,0.5) {};
\node [fill=blue,circle,scale=0.5] at (3.5,1.35) {};

\draw [blue,thick]  (-4,-1.75) rectangle (-2.5,-2.25);
\draw [blue,thick]   (-2.5,-1.75) rectangle (-1,-2.25);
\draw [red,thick]   (-1,-1.75) rectangle (0.5,-2.25);
\draw [blue,thick]   (0.5,-1.75) rectangle (2,-2.25);
\draw [blue,thick]   (2,-1.75) rectangle (3.5,-2.25);
\node at (-3.25,-2) {Correct};
\node at (-1.75,-2) {Correct};
\node at (-0.25,-2) {Incorrect};
\node at (1.25,-2) {Correct};
\node at (2.75,-2) {Correct};
\end{tikzpicture}
\caption{Illustration of the plant state, $x_t$, and the correctness of the codeword decoding, when $n=3$.}
\label{fig:illu-process}
\end{figure}
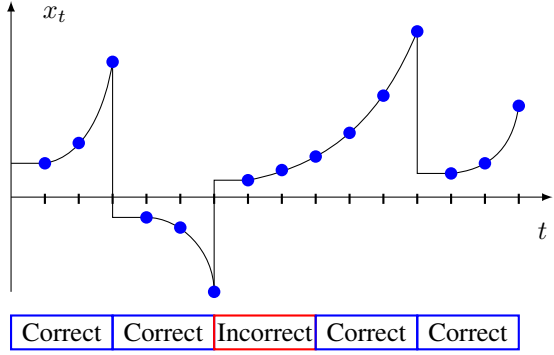

\subsection{Construction of a Zooming Quantizer} \label{sec:quantizer}

We consider a  value-independent zooming quantizer with two construction steps:

First, the $2^{nR}$-step uniform quantizer with the range $(-H,H)$ is defined as
\begin{equation} \label{quantizer}
\myQQ_H(x)=
\left\lbrace
\begin{aligned}
&\left(k-\frac{1}{2}\right) \varDelta, && \text{if } x \in \left[(k-1)\varDelta,k\varDelta  \right)\\
&\left(-k+\frac{1}{2}\right) \varDelta, && \text{if } x \in \left[-k\varDelta,(-k+1)\varDelta  \right)\\
& 0, &&\text{else}\\
\end{aligned}
\right.
\end{equation}
where $\varDelta = H / 2^{nR-1}$ and $k = 1,2,...,2^{nR-1}$.

Second, the quantization-range updating rule is defined as
\begin{equation} \label{update rule}
H_{t+1} = 
\left\lbrace
\begin{aligned}
&L H_t, && t= k n, k \in \mathbb{Z}_+, s_t=0\\
&\max\left\lbrace \Xi(0), \frac{1}{L} H_t \right\rbrace, && t= k n, k \in \mathbb{Z}_+, s_t=1\\
& H_t, && \text{else},
\end{aligned}
\right. 
\end{equation}
where $(-\Xi(0),\Xi(0))$ is the minimum quantization range, and $L$ is the scaling factor of the quantization-range updating rule. 
In other words, the quantization range enlarges and shrinks when incorrect and correct control codewords are detected, respectively.
Assuming that $\Xi(0)$, $L$ and the initial quantization range $(-H_0,H_0)$ are known by both the controller and the actuator, they can update the quantizer in a synchronous manner based on the shared information $s_t,\forall t$.

Therefore, this zooming quantizer is defined as $\myQQ_{H_t}(\cdot)$.
Note that the pre-quantized control signal $\tilde{\mu}_t, \forall t = kn$, is quantized by $\myQQ_{H_{t}}(\cdot)$ with the quantization range $(-H_{t},H_{t})$.

Furthermore, it is easy to see that the quantization-range updating forms a Markov chain with the states $\Xi(i) \triangleq L^{i} \Xi(0), \forall i \in \mathbb{Z}_0$, and the state transition probability is defined as
\begin{equation}
\myprobability{H_{t+n} = \Xi(j) \vert H_t = \Xi(i)} = 
\left\lbrace
\begin{aligned}
& 1-\epsilon ,&& i\geq 0,j=[i-1]^+\\
& \epsilon ,&& i\geq 0,j=i+1\\
& 0 ,&& \text{else}\\
\end{aligned}
\right.
\end{equation}
where $[z]^+ \triangleq \max\{z,0\}$.
The stationary distribution of the Markov chain can be calculated as
\begin{equation} \label{Phi}
\Phi_i\triangleq \myprobability{H_t = \Xi(i)} = \left(\frac{\epsilon}{1-\epsilon}\right)^{i}\frac{1-2 \epsilon}{1-\epsilon}.
\end{equation}

In what follows, we analyze the condition of the existence of the zooming-quantizer such that the control signal before quantization, $\tilde{\mu}_t$, is always within the quantization range $(-H_t,H_t)$, i.e.,
\begin{equation} \label{condition:signal}
-H_t <\tilde{\mu}_t < H_t.
\end{equation}
and
\begin{equation}\label{condition:step}
-\frac{H_t}{2^{nR-2}} < v_t < \frac{H_t}{2^{nR-2}}.
\end{equation}

\subsection{The Existence of a Zooming Quantizer}\label{sec:analysis}
Taking \eqref{command} and \eqref{action} into \eqref{system}, the plant state at $t+n$ can be written as
\begin{equation} \label{x_t+n}
\begin{aligned}
{x}_{t+n}  &= a^{n-1} \left(a x_t + b u_t +w_{t} \right) + a^{n-2}w_{t+1} + \cdots + w_{t+n}\\
&= a^{n-1} \left(a x_t + b u_t \right) + \sum_{i=0}^{n-1} a^{n-1-i}w_{t+i}\\
&\triangleq \hat{x}_{t+n} + \omega_{t+n},
\end{aligned}
\end{equation}
where
\begin{equation}
-\omega_{\max}<\omega_{t+n} < \omega_{\max},
\end{equation}
and 
\begin{equation}
\omega_{\max} \triangleq  w_{\max}\sum_{i=0}^{n-1} a^{n-i}
= w_{\max} \frac{a^n-1}{a-1}.
\end{equation}

Based on \eqref{command} and \eqref{x_t+n}, we have
\begin{equation} \label{estimation}
\hat{x}_{t+n} = 
\left\lbrace
\begin{aligned}
&a^n x_0 ,&&\hspace{-0.3cm} t =0\\
&
\begin{aligned}
&a^n \left(x_t - \hat{x}_t\right) + a^{n-1} b v_t \\
&= a^n \omega_{t} + a^{n-1}b v_t,
\end{aligned} &&\hspace{-0.3cm} t= k n, k \in \mathbb{Z}_+, s_t = 1\\
&\begin{aligned}
&a^n x_t =\!a^n \left(\!x_t - \hat{x}_t\!\right) \!+\! a^n  \hat{x}_t\\
&=a^n \omega_{t} + a^{n-1}b \left(\frac{a}{b}  \hat{x}_t\right)
,
\end{aligned} &&\hspace{-0.3cm}t= k n, k \in \mathbb{Z}_+, s_t = 0 \\
\end{aligned}
\right.
\end{equation}

Furthermore, to satisfy the quantization-range updating rule \eqref{update rule}, we have the following conditions:
\begin{align}
\left\lbrace
\begin{aligned}
&\text{If }\!t = 0, -\frac{a}{b} \hat{x}_{t+n} \in \left(-H_0, H_0\right)\!;\\
&\text{If }\! t= k n, k \in \mathbb{Z}_+, s_t=0, \  -\frac{a}{b}\hat{x}_{t+n} \in \left(-L H_t, L H_t\right)\!; \\
&\text{If }\! t\!=\! k n, k \in \mathbb{Z}_+, s_t\!=\!1\text{ and }H_t \!>\! \Xi(0),  \!-\!\frac{a}{b}\hat{x}_{t+n} \!\in \!\left(\!-\!\frac{H_t}{L} ,\! \frac{H_t}{L}\! \right)\! ;\\
&\text{If }\! t\!=\! k n, k \in \mathbb{Z}_+, s_t\!=\!1\text{ and }H_t \!=\! \Xi(0), \!-\!\frac{a}{b}\hat{x}_{t+n} \!\in \!\left(\!-\Xi(0),\! \Xi(0)\!\right)\!;
\end{aligned}
\right.
\end{align}
which can be simplified based on \eqref{estimation} as
\begin{align} \label{condition_Q}
\left\lbrace
\begin{aligned}
& H_0 > \frac{a^{n+1}}{b} x_{0,\max}\\
&\left(L - a^n\right)\Xi(0) > \frac{a^{n+1}}{b} \frac{a^n-1}{a-1} w_{\max} \\
&\left(1 - L \frac{a^{n}}{2^{nR-2}}\right) \Xi(0) > \frac{a^{n+1}}{b} \frac{a^n-1}{a-1} w_{\max} .
\end{aligned}
\right.
\end{align}
Without loss of generality, the initial quantization range can be set as the minimum range in \eqref{update rule} that satisfies the first condition in \eqref{condition_Q}, i.e.,
\begin{equation}
H_0 = \Xi\left(\left\lceil \log_L\left(\frac{a^{n+1} x_{0,\max}}{b \Xi(0)}\right) \right \rceil\right).
\end{equation}

Therefore, the condition \eqref{condition_Q} can be satisfied as long as $\Xi(0)$ is sufficiently large and there exists $L$ such that
\begin{equation}\label{quan_condition}
a^n < L < \frac{2^{nR-2}}{a^n}.
\end{equation}
After simplification, we have the following result.
\begin{proposition} \label{prop:quantizer}
	\normalfont
The necessary and sufficient condition of the existence of a zooming quantizer $\myQQ_{H_t}(\cdot)$ for which the pre-quantized signal is always within the quantization range, is given by
\begin{equation} \label{condition_Q_final}
n > \frac{2}{R-2 \log_2 a},\text{ and } R > 2 \log_2 a.
\end{equation}
\end{proposition}

Although the condition introduced by the zooming-quantizer in Proposition~\ref{prop:quantizer} is independent of the error probability, the stability condition still depends on it as discussed in the following section.

\section{Stability and Performance Analysis of the Practical WNCS}
\subsection{Stability Condition and Performance Analysis}
Before we proceed further, we need the following definitions/properties:
\begin{enumerate}
	\item We divide the time into patterns each consisting of $n$ consecutive time slots. Pattern $j$, $j\in \mathbb{Z}_0$, is defined as the $(j+1)$~th consecutive $n$ time slots after the previous correct decoding of a codeword, as illustrated in Fig.~\ref{fig:patterns}. 
	\item The patterns forms an infinite state discrete Markov chain, and the stationary probability of pattern $j$ can be derived~as 
	\begin{equation}
	\Psi_j\triangleq 
	\myprobability{\patt_{{kn}}=j} 
	 = (1-\epsilon)\epsilon^{j}, \forall j\in \mathbb{Z}_0.
	\end{equation}
\end{enumerate}

Since the last state in each patten has the largest variance, to analyze the mean-square stability condition~\eqref{first_condition}, we only need to analyze the down-sampled {random process} $X_{kn}, \forall k \gg 1$.

\begin{figure}
	\centering
\begin{tikzpicture}[scale = 0.7]

\draw [blue] (-4.5,2) rectangle (-2.5,1.5);

\draw [red] (-4.5,1) rectangle (-2.5,0.5);
\draw  [blue] (-2.5,1) rectangle (-0.5,0.5);

\draw [red] (-4.5,0) rectangle (-2.5,-0.5);
\draw [red] (-2.5,0) rectangle (-0.5,-0.5);
\draw  [blue] (-0.5,0) rectangle (1.5,-0.5);

\draw [red] (-4.5,-1) rectangle (-2.5,-1.5);
\draw [red] (-2.5,-1) rectangle (-0.5,-1.5);
\draw [red] (-0.5,-1) rectangle (1.5,-1.5);
\draw  [blue] (1.5,-1) rectangle (3.5,-1.5);
\node at (-3.5,1.75) {Correct};

\node at (2.5,-1.25) {Correct};
\node at (0.5,-1.25) {Incorrect};
\node at (0.5,-0.25) {Correct};
\node at (-1.5,-1.25) {Incorrect};
\node at (-1.5,-0.25) {Incorrect};
\node at (-1.5,0.75) {Correct};
\node at (-3.5,-1.25) {Incorrect};
\node at (-3.5,-0.25) {Incorrect};
\node at (-3.5,0.75) {Incorrect};
\draw  [blue] (-6.5,2) rectangle (-4.5,1.5);
\draw  [blue] (-6.5,1) rectangle (-4.5,0.5);
\draw  [blue] (-6.5,0) rectangle (-4.5,-0.5);
\draw  [blue] (-6.5,-1) rectangle (-4.5,-1.5);
\node at (-5.5,1.75) {Correct};
\node at (-5.5,0.75) {Correct};
\node at (-5.5,-0.25) {Correct};
\node at (-5.5,-1.25) {Correct};
\draw [dashed] (-4.5,2.5) -- (-4.5,-2.5);
\draw [dashed] (-2.5,2.5) -- (-2.5,-2.5);
\draw [dashed] (-0.5,2.5) -- (-0.5,-2.5);
\draw [dashed] (1.5,2.5) -- (1.5,-2.5);
\node at (4,-1.25) {...};
\node at (-7,-1.25) {...};
\node at (-3.5,-2) {Pattern $0$};
\node at (-1.5,-2) {Pattern $1$};
\node at (0.5,-2) {Pattern $2$};
\end{tikzpicture}
\caption{Illustration of the patterns.}
\label{fig:patterns}
\end{figure}
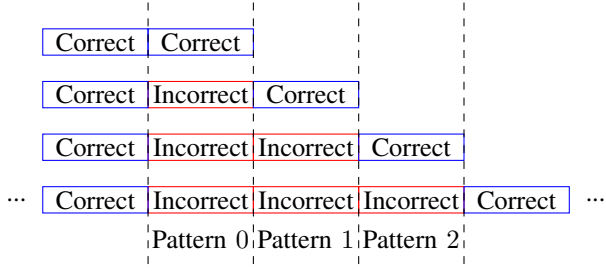

Based on \eqref{x_t+n} and \eqref{estimation}, 
the state $X_{kn}$ given that $H_{kn}=\Xi(i)$ and $\patt_{kn}=j$, can be written~as
\begin{equation}
\begin{aligned}
X_{kn} = &a^{j n-1}b V_{kn-j n} + a^{(j+1)n-1} W_{kn-(j+1)n}\\
&+ a^{(j+1)n-2} W_{kn-(j+1)n +1}
+\cdots + W_{kn-1},
\end{aligned}
\end{equation}
where $V_{kn-(j+1)n}$ is the quantization noise caused by the zooming quantizer with the range $(-\Xi(i),\Xi(i))$, and hence
\begin{equation}
\vert V_{kn-j n} \vert <v_{i,\max} \triangleq \frac{L^i \Xi(0)}{2^{nR-2}}.
\end{equation}
Since $W_{i},\forall i$ is an i.i.d. random process, the conditional expectation of $\mynorm{X_{kn}}$ is given by
\begin{equation}
\begin{aligned}
&\myexpect{}{\mynorm{X_{kn}} \vert H_{kn-j n}=\Xi(i),\patt_{kn}=j,V_{kn-(j+1)n}}\\
&= \left(V_{kn-j n}\right)^2 \tr{\mynorm{a^{j n-1}b}}
+\sigma^2_W 
\sum_{m=0}^{(j+1)n-1}  \tr{\mynorm{a^{m}}} .
\end{aligned}
\end{equation}

Thus, by defining the conditional expectation
\begin{equation}
\myA_{j,i} \triangleq \myexpect{}{\mynorm{X_{kn}} \vert H_{kn-jn}=\Xi(i),\patt_{kn}=j},
\end{equation}
we have
\begin{equation}
\begin{aligned}
\sigma^2_W &
\sum_{m=0}^{(j+1)n-1}  \tr{\mynorm{a^{m}}} 
<\myA_{j,i} \\
&<\left(v_{i,\max}\right)^2
 \tr{\mynorm{a^{j n-1}b}}
+\sigma^2_W 
\sum_{m=0}^{(j+1) n-1}  \tr{\mynorm{a^{m}}}, 
\end{aligned}
\end{equation}
since 
$
0< \myexpect{\!}{\left(V_{kn-(j+1)n}\right)^2} < \left(v_{i,\max}\right)^2.
$

By the law of total expectation, we have 
	\begin{equation}
	\begin{aligned}
	\myA 
	&=\sum_{i=0}^{\infty} \sum_{j=0}^{\infty}  \Phi_i \Psi_j  \myA_{j,i}
	= \sum_{i=0}^{\infty} \Phi_i \sum_{j=0}^{\infty}  \Psi_j \myA_{j,i},
	\end{aligned}
	\end{equation}
where $	\Phi_i = \myprobability{H_{kn}=\Xi(i)}$ is defined in \eqref{Phi}, and the first equality is due to the fact that 
\begin{equation}
\myprobability{{H}_{kn-jn}=\Xi(i) \vert \patt_{{kn}}=j}
=
\myprobability{H_{kn-jn}=\Xi(i)}.
\end{equation}
	
After simplification, the average cost can be bounded as
\begin{equation} \label{upper_bound}
\begin{aligned}
\myA 
&<\sum_{i=0}^{\infty} \Phi_i 
{(1-\epsilon)} \times\\
&\left(
{v^2_{i,\max}} b^2  \sum_{j=0}^{\infty} \epsilon^j \left(a^2\right)^{j n-1}
+
\sigma^2_W \sum_{j=0}^{\infty} \epsilon^j \frac{\left(a^2\right)^{(j+1)n}-1}{a^2-1}
\right)
\end{aligned}
\end{equation}
and
\begin{equation} \label{lower_bound}
\begin{aligned}
\myA 
>
{(1-\epsilon)}\left(
\sigma^2_W \sum_{j=0}^{\infty} \epsilon^j \frac{\left(a^2\right)^{(j+1)n}-1}{a^2-1}
\right).
\end{aligned}
\end{equation}

Therefore, a sufficient and a necessary condition to make $\myA$ bounded can be obtained by making the upper and the lower bound of $\myA$ bounded, respectively, which are given by 
\begin{equation} \label{condition_1}
\epsilon < \min \left\lbrace \frac{1}{1+ L^2}, \frac{1}{a^{2n}} \right\rbrace,
\end{equation}
and
\begin{equation} \label{condition_2}
\epsilon <  \frac{1}{a^{2n}}.
\end{equation}
Therefore, by integrating the conditions~\eqref{quan_condition}, \eqref{condition_1} and \eqref{condition_2}, we have the following result.
\begin{theorem}	 \label{theo:2}
	\normalfont	
A sufficient condition and a necessary condition to stabilize the plant with the control policy in Sec.~\ref{sec:control} and the zooming quantizer $\myQQ_{H_t}(\cdot)$ defined in Sec.~\ref{sec:quantizer} are given~by
\begin{equation} \label{suffi}
n > \frac{2}{R-2 \log_2 a},R > 2 \log_2 a, \epsilon < \frac{1}{1+a^{2n}},
\end{equation}
and
\begin{equation}
n > \frac{2}{R-2 \log_2 a},R > 2 \log_2 a, \epsilon < \frac{1}{a^{2n}}.
\end{equation}
\end{theorem}
From Theorem~\ref{theo:2}, we see that the sufficient and necessary conditions are close to each other especially when $n \gg 1$.
In general, we need a large blocklength, a high data rate and a small error probability to stabilize the system.

Taking \eqref{Yury_e} into Theorem~\ref{theo:2}, the sufficient and necessary stability regions in terms of $(n,R)$ can be obtained. Similar to Proposition~\ref{prop:bound_nair}, we can obtain the following property of the stability region.
\begin{proposition}
	\normalfont	
	Assume the stability region exists based on Theorem~\ref{theo:2}. If $2 \log_2 a \geq C-2 \sqrt{\nu\ln a}$, the two-dimensional stability region is $n$-bounded; otherwise, the region is $n$-unbounded when $R \in \left(2 \log_2 a, C- 2 \sqrt{\nu \ln a} \right)$.	
\end{proposition}

Fig.~\ref{fig:region_compare} plots the stability region of the ideal WNCS in Corollary~\ref{theo:1} (red lines), where the quantizer utilizes historical information for control signal quantization, and necessary stability region of the practical WNCS with the zooming quantizer (blue circled lines). We see that the stability region of the ideal WNCS is larger than that of the practical WNCS, due to a larger quantization complexity.
\begin{figure}[t]
	\renewcommand{\captionfont}{\small} \renewcommand{\captionlabelfont}{\small}
	\centering
	\includegraphics[scale=0.65]{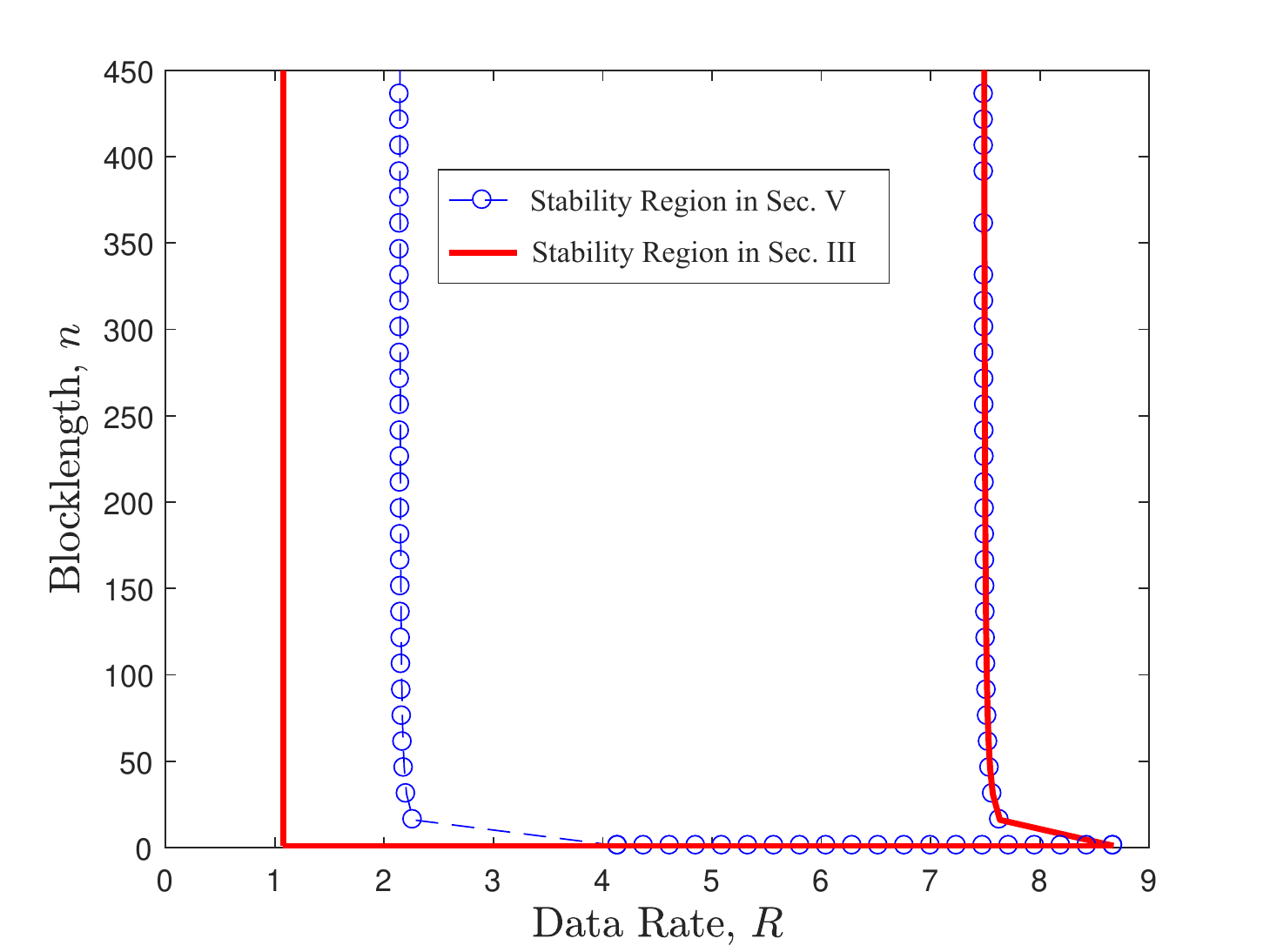}
	\caption{Comparison of stability regions achieved by different control and quantization policies with $\snr=30$~dB and $a=2.1$.}
	\label{fig:region_compare}
\end{figure}

After simplification of the infinite series in \eqref{upper_bound} and \eqref{lower_bound}, we have the following result.
\begin{proposition}\label{prop：bounds}
	\normalfont	
If the sufficient condition~\eqref{suffi} is satisfied, the upper and lower bounds of the average cost are given by
	\begin{equation}\label{upper}
\begin{aligned}
		\myA <& \frac{1}{1- \frac{\epsilon}{1-\epsilon} L^2}  (1-2 \epsilon)	
	\left(\frac{ \Xi(0)}{2^{nR-2}}\right)^2  b^2  \frac{a^{-2}}{1-\epsilon a^{2n}}\\
	&+
	(1-\epsilon)\frac{\sigma^2_W}{a^2-1} \left(\frac{a^{2n}}{1-\epsilon a^{2n}} - \frac{1}{1-\epsilon}\right),
\end{aligned}
	\end{equation}
	and
	\begin{equation}\label{lower}
	\myA > (1-\epsilon)\frac{\sigma^2_W}{a^2-1} \left(\frac{a^{2n}}{1-\epsilon a^{2n}} - \frac{1}{1-\epsilon}\right).
	\end{equation}	
\end{proposition}
\begin{remark}
	The lower bound \eqref{lower} shows that the average cost grows unbounded as $n$ increases, which explains the counter-intuitive finding in Sec.~\ref{sec:conter} that
although the plant can be stabilized with an arbitrarily large blocklength in some cases, the average cost of the system is very high, which is not desirable.
\end{remark}

\subsection{Numerical Results}
By representing the decoding-error probability $\epsilon$ with the data rate $R$ and the blocklength $n$ in Proposition~\ref{prop：bounds}, i.e., taking \eqref{Yury_e} into \eqref{upper} and \eqref{lower}, we numerically present the upper and lower bounds of the average cost function in terms of $R$ and~$n$.
Unless otherwise stated,
we set $\snr=3$~dB, $a=1.01$, $\Xi(0)=10^5$, $L = 10^2$ and $\sigma^2_W = 10^{-10}$~\cite{LiuJIoT}.

{For the configuration of the data rate $R$, the blocklength (latency) $n$, and the reliability $\epsilon$, we follow the pioneer work of finite blocklength information theory~$[22]$, where $R\in [0,3.5]$, $n\in [1,1000]$, and $\epsilon \in [10^{-6}, 10^{-3}]$. Note that in practice $R$, $n$ and $\epsilon$ are determined by the specified channel encoding and modulation scheme.
	For example, considering BCH$(255, 63)$ code with the BPSK modulation scheme, the input number and the output number of bits of the encoder are $63$ and $255$, respectively, and each symbol contains $2$ bits. The channel blocklength is equal to $255/2\approx 128$, and the data rate is equal to $63/128 \approx 0.5$. The detection error probability can be obtained by simulations. 
}

In Fig.~\ref{fig:n}, we plot the upper and lower bounds of the average cost function with different $n$. We see that the upper bound of the average cost depends on $R$ and decreases first and then increases with $n$, while the lower bound is independent of $R$ and monotonically increases with $n$.
Also, it can be observed that the upper and lower bounds are tight when $n$ is large, e.g., $n>200$, while the gap between them is large when $n$ is small.
We have also obtained the simulation result of the average cost with $R=0.19$ based on \eqref{cost}, where $t = 10^8$ and the expectation is obtained by Monte Carlo simulation with $10^4$ trials.
It is interesting to see that the analytical upper bound is close to the simulation result. Thus, in the following, we only evaluate and illustrate the upper bound of the average cost.

In Fig~\ref{fig:R}, we plot the upper bound of the average cost with different $R$. We see that the average cost is not sensitive to the data rate in the range of $(0.25,1)$, and it grows up quickly with a decreasing or increasing  $R$ outside of the range.
In Fig.~\ref{fig:performance}, we give a contour plot of the upper bound  of the average cost in terms of $R$ and $n$.
It can be observed that the average cost function is flat when $R\in(0.25,1)$ and $n \in (100,250)$.
Using \eqref{upper}, we can numerically minimize the average cost by optimizing the communication parameters, where the optimal parameters are $n^\star = 120 $ and $R^\star = 0.7$.

\begin{figure}[t]
	\renewcommand{\captionfont}{\small} \renewcommand{\captionlabelfont}{\small}
	\centering
	\vspace{0.4cm}
	\includegraphics[scale=0.64]{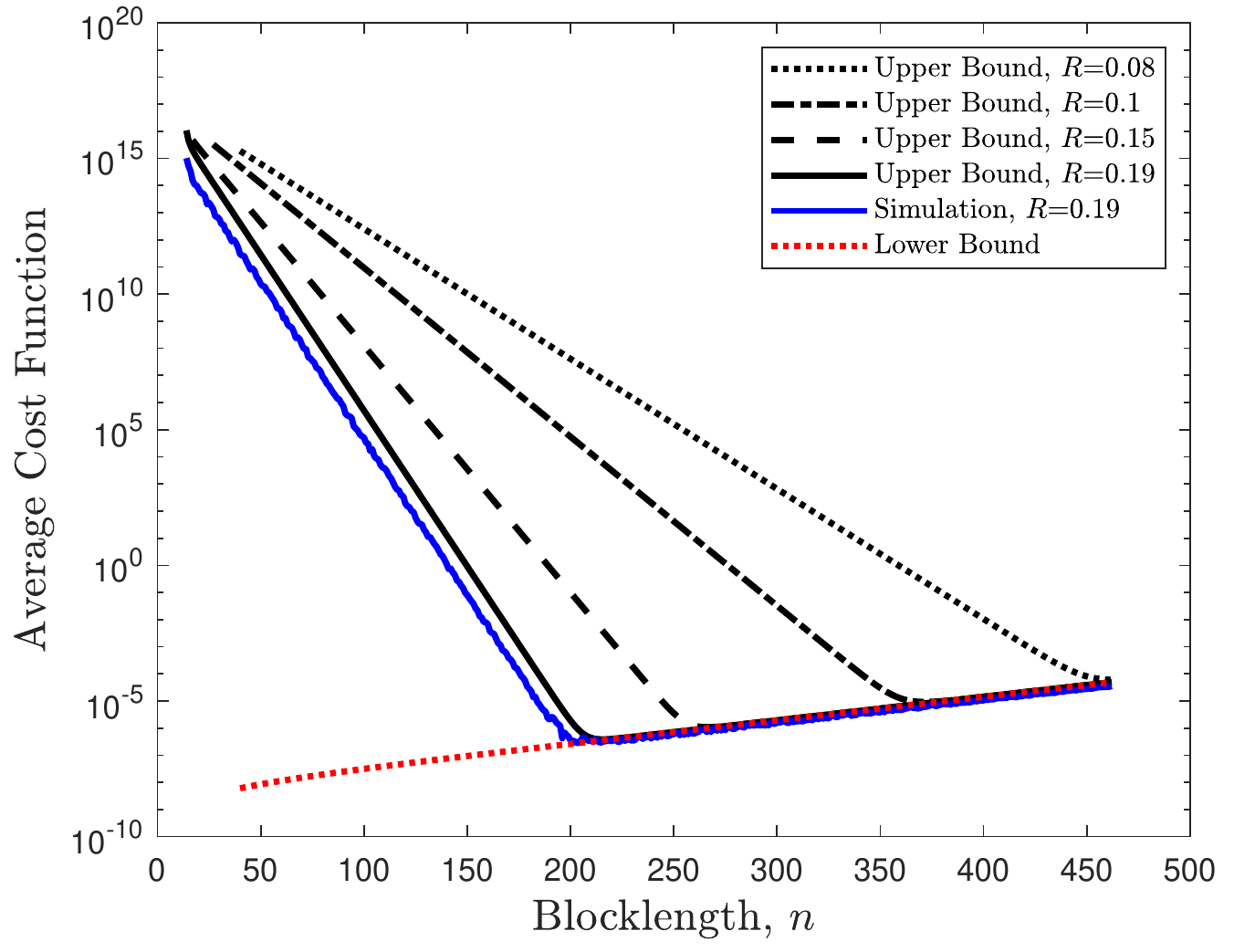}
	\caption{Average cost function versus blocklength $n$.}
	\label{fig:n}
\end{figure}

\begin{figure}[t]
	\renewcommand{\captionfont}{\small} \renewcommand{\captionlabelfont}{\small}
	\centering
	\includegraphics[scale=0.66]{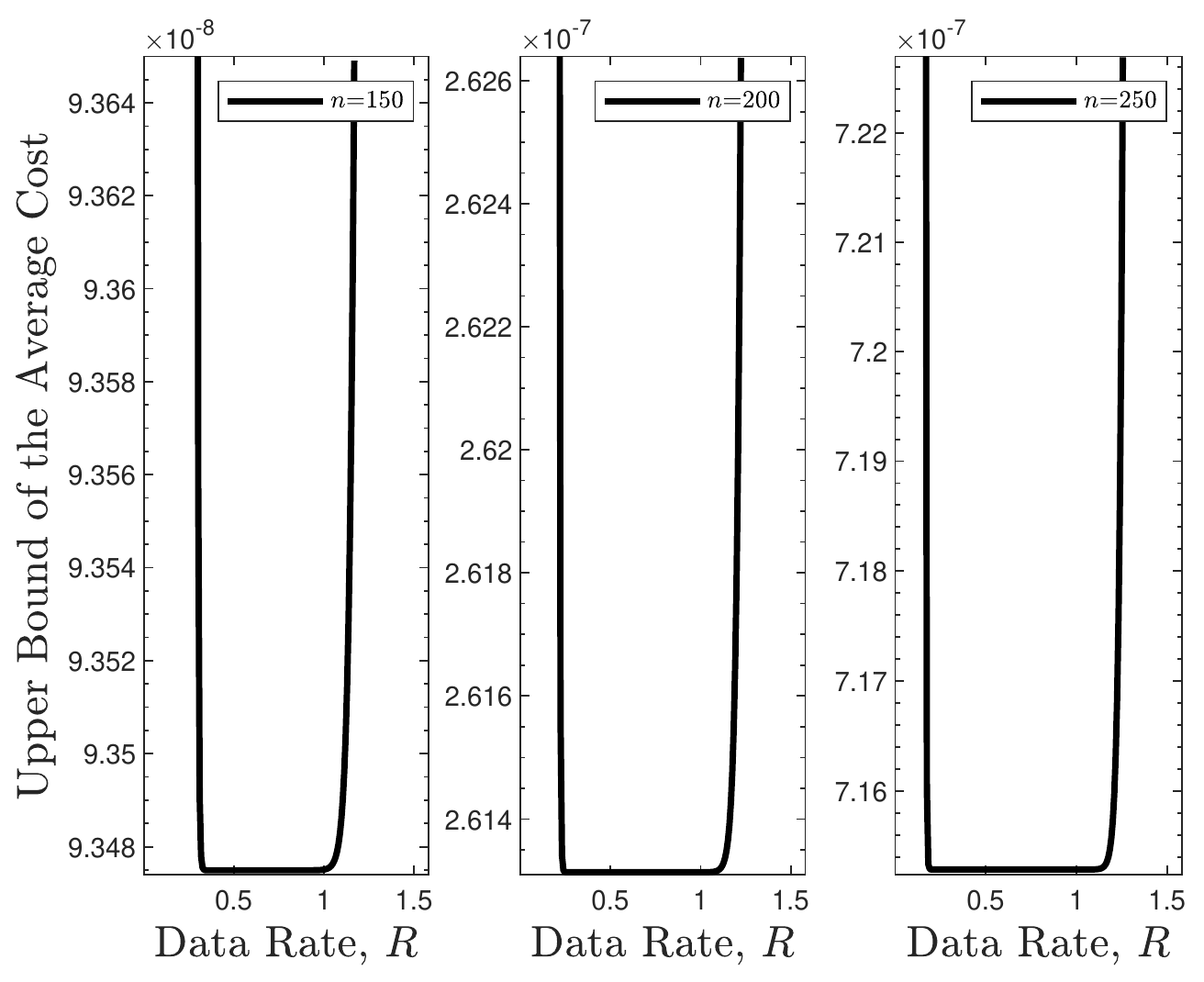}
	\caption{Average cost function versus data rate $R$.}
	\label{fig:R}
\end{figure}

\begin{figure}[t]
	\renewcommand{\captionfont}{\small} \renewcommand{\captionlabelfont}{\small}
	\centering
	\includegraphics[scale=0.68]{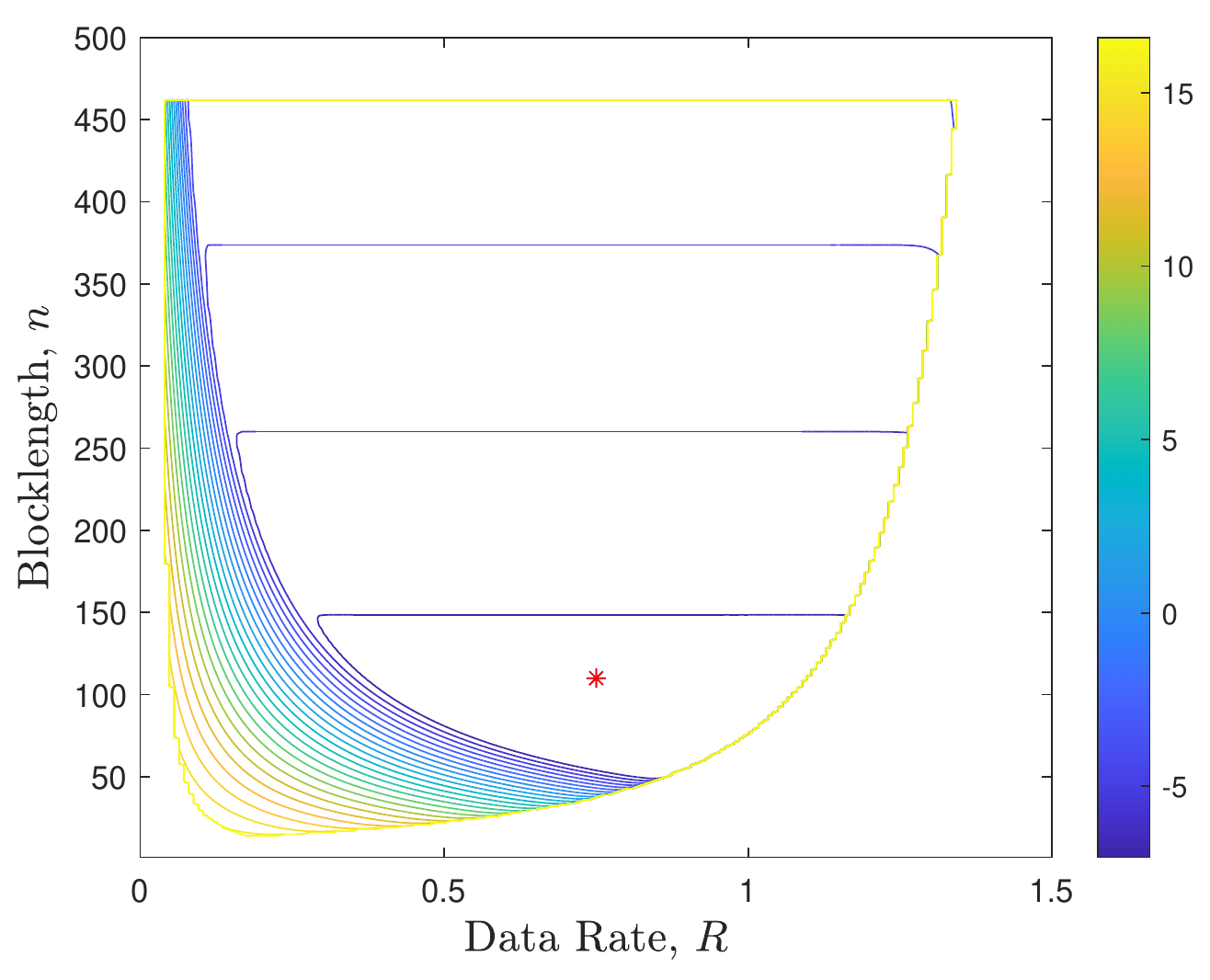}
	\caption{A contour plot of the upperbound of the average cost (in a base-$10$  logarithmic scale), $\myA$, versus the blocklength $n$ and the data rate $R$.}
	\label{fig:performance}
\end{figure}

\section{Conclusions}
\noindent
In this paper, we have investigated a WNCS with a detailed model of the wireless communication system, which jointly considers the inter-related communication parameters.
We have derived the stability region of the WNCS in terms of the communication parameters, and have obtained a necessary and sufficient condition under which the stability region is $n$-bounded.
This result shows an important counter-intuitive finding that the plant can be stabilized with an arbitrarily large blocklength (i.e., a large latency) as long as the SNR of the wireless channel is greater than a certain value determined by the control-system parameters.
We have also analyzed the average cost function of the WNCS and derived tight upper and lower bounds on the average cost function in terms of the communication parameters.   
The performance analysis explains the above mentioned counter-intuitive finding: even though a plant can be stabilized with an arbitrarily large $n$, the average cost of the system is arbitrarily large and is not desirable.
Our analytical results also enable near-optimal design of the communication parameters to minimize the average cost within the stability region.
Of interest for future work is the vector plant systems and wireless control through fading channels~\cite{liu2020remote}.
{Furthermore, we will consider the energy consumption of the controller for transmission and that of the actuator for control. Thus, we will have a multi-object optimization problem with three object functions: the average cost of the plant, the energy cost of the controller and the energy cost of the actuator. 
This is can be viewed as a Pareto-optimal problem. 
When a parameter is incorporated in a multi-component objective function, the search-space is inevitably expanded, which may result in potentially NP-hard problems. 
To achieve the Pareto optimum, we will consider both traditional mathematical algorithms and machine learning algorithms to find all Pareto-optimal operating points, where none of the above-mentioned conflicting objects can be improved without degrading at least one of others~\cite{Pareto}.}


\section*{Appendix A: Proof of Lemma~\ref{lemma:1}}
\noindent
We note that proof follows the same steps as in \cite{Nair09} and we only present a sketch of proof.
The proof of the necessity is based on the maximum entropy theorem \cite{BookInfo}.
The proof of the sufficiency is based on a constructive scheme with the following elements:
1) Both the controller and the actuator adopt an adaptive quantizer, which requires large computation and storage resources at both the controller and actuator.
2) The actuator and controller share a state estimator of the plant that estimates the state as $\hat{x}_t$, $\forall t \in \mathbb{Z}_0$, based on the past correctly decoded codewords.
3) The actuator generates the control action as ${u}_t = \kappa \hat{x}_t$, where $\kappa$ is a constant such that $\vert a  + \kappa b \vert <1$.
4) The controller maps the index number of the quantized $(x_t-\hat{x}_t)$ into an $(n,R,\epsilon)$ codeword and transmits it to the actuator.


We also would like to clarify that in \cite{Nair09}, it is assumed that the transmitter knows the available error-free data rate of the channel before each transmission, while in our system, the codeword transmission is not error-free.
Although \cite{Nair09} and our scenario are not exactly the same at the first glance, the analysis is the same. 
Particularly, in \cite{Nair09}, the transmitter sends no information when its available data rate is $0$.
In our scenario, the transmitter always sends information with rate $R$.
If the detection is successful, it is an error-free transmission with rate $R$.
If the detection is failed and no useful information is transmitted, it is equivalent to the scenario with previously known zero transmission data rate. Thus, in our system, with or without unit delay of the available data-rate information are equivalent in this sense.
Therefore, Lemma~1 can be obtained from Theorem 4.1 of \cite{Nair09} by considering two possible data rates, i.e., $R$ and $0$, with the probability of $(1-\epsilon)$ and $\epsilon$, respectively.

\section*{Appendix B: Proof of Proposition~\ref{prop:bound_nair}}
\noindent
Let $R\in (0,C)$. 
Using the asymptotic approximation of the Q function that
\begin{equation}
Q\left(x\right) \approx
\frac{1}{2}
\frac{e^{-x^2/2}}{x\sqrt{\pi/2}}
\sum_{i=0}^{\infty} (-1)^i \frac{(2i-1)!!}{x^{2i}},
\end{equation}
we have 
\begin{equation}
\begin{aligned}
\lim\limits_{n \rightarrow \infty} 
Q\left(\sqrt{\frac{n}{\nu}}(C-R)\right)
&\approx
\frac{1}{2} \frac{\exp\left(-\frac{1}{2} 
	{\frac{n}{\nu}}(C-R)^2
	\right)}{\sqrt{\frac{n}{\nu}}(C-R) \sqrt{\pi/2}}\\
&= \frac{\left(\exp\left(- 
	{\frac{1}{2\nu}}(C-R)^2
	\right)\right)^n}{\sqrt{ \frac{ 2\pi }{\nu}}(C-R) \sqrt{ n }}.
\end{aligned}
\end{equation}
Therefore, the first term in Corollary~\ref{theo:1}, $Q\left(\sqrt{\frac{n}{\nu}}(C-R)\right) a^{2n}$, is bounded (and goes to zero) when $n \rightarrow \infty$, iff $a^2 \leq \exp\left(
{\frac{1}{2\nu}}(C-R)^2
\right)$, i.e., 
\begin{equation}\label{proof1}
R \leq C- 2\sqrt{\nu \ln a}.
\end{equation}

Also, since $Q\left(\sqrt{\frac{n}{\nu}}(C-R)\right) \rightarrow 0$ when $n \rightarrow \infty$,
the second term in Corollary~\ref{theo:1}, i.e., $\left(1 - Q\left(\sqrt{\frac{n}{\nu}}(C-R)\right)\right) \frac{a^{2n}}{2^{2n R}}$, is less than one when $n \rightarrow \infty$, iff $a \leq 2^R$. The term goes to zero when $n \rightarrow \infty$, iff $a < 2^R$, i.e., 
\begin{equation} \label{proof2}
R >  \log_2 a.
\end{equation}
It is straightforward that when the above equality holds the plant cannot be stabilized based on Corollary~\ref{theo:1}.

Therefore, taking \eqref{proof1} and \eqref{proof2} into Corollary~\ref{theo:1} yields Proposition~\ref{prop:bound_nair}.

\balance

\ifCLASSOPTIONcaptionsoff
\fi
%

\begin{thebibliography}{10}
	\providecommand{\url}[1]{#1}
	\csname url@samestyle\endcsname
	\providecommand{\newblock}{\relax}
	\providecommand{\bibinfo}[2]{#2}
	\providecommand{\BIBentrySTDinterwordspacing}{\spaceskip=0pt\relax}
	\providecommand{\BIBentryALTinterwordstretchfactor}{4}
	\providecommand{\BIBentryALTinterwordspacing}{\spaceskip=\fontdimen2\font plus
		\BIBentryALTinterwordstretchfactor\fontdimen3\font minus
		\fontdimen4\font\relax}
	\providecommand{\BIBforeignlanguage}[2]{{%
			\expandafter\ifx\csname l@#1\endcsname\relax
			\typeout{** WARNING: IEEEtran.bst: No hyphenation pattern has been}%
			\typeout{** loaded for the language `#1'. Using the pattern for}%
			\typeout{** the default language instead.}%
			\else
			\language=\csname l@#1\endcsname
			\fi
			#2}}
	\providecommand{\BIBdecl}{\relax}
	\BIBdecl
	
	\bibitem{ParkSurvey}
	P.~Park, S.~C. Ergen, C.~Fischione, C.~Lu, and K.~H. Johansson, ``Wireless
	network design for control systems: A survey,'' \emph{IEEE Commun. Surveys
		Tuts.}, vol.~20, no.~2, pp. 978--1013, Second Quarter 2018.
	
	\bibitem{Turkey}
	Y.~Sadi and S.~C. Ergen, ``Joint optimization of wireless network energy
	consumption and control system performance in wireless networked control
	systems,'' \emph{IEEE Trans. Wireless Commun.}, vol.~16, no.~4, pp.
	2235--2248, Apr. 2017.
	
	\bibitem{Schulz}
	P.~{Schulz}, M.~{Matthe}, H.~{Klessig}, M.~{Simsek}, G.~{Fettweis},
	J.~{Ansari}, S.~A. {Ashraf}, B.~{Almeroth}, J.~{Voigt}, I.~{Riedel},
	A.~{Puschmann}, A.~{Mitschele-Thiel}, M.~{Muller}, T.~{Elste}, and
	M.~{Windisch}, ``Latency critical {IoT} applications in {5G}: Perspective on
	the design of radio interface and network architecture,'' \emph{IEEE Commun.
		Mag.}, vol.~55, no.~2, pp. 70--78, Feb. 2017.
	
	\bibitem{Splitting1}
	W.~{Liu}, X.~{Zhou}, S.~{Durrani}, and P.~{Popovski}, ``A novel receiver design
	with joint coherent and non-coherent processing,'' \emph{IEEE Trans.
		Commun.}, vol.~65, no.~8, pp. 3479--3493, Aug. 2017.
	
	\bibitem{Splitting2}
	Y.~{Wang}, W.~{Liu}, X.~{Zhou}, and G.~{Liu}, ``On the performance of splitting
	receiver with joint coherent and non-coherent processing,'' \emph{IEEE Trans.
		Signal Process.}, vol.~68, pp. 917--930, 2020.
	
	\bibitem{Jiang1}
	D.~{Jiang}, Y.~{Wang}, Z.~{Lv}, W.~{Wang}, and H.~{Wang}, ``An energy-efficient
	networking approach in cloud services for iiot networks,'' \emph{IEEE J. Sel.
		Areas Commun.}, vol.~38, no.~5, pp. 928--941, 2020.
	
	\bibitem{Jiang2}
	D.~{Jiang}, Y.~{Wang}, Z.~{Lv}, S.~{Qi}, and S.~{Singh}, ``Big data analysis
	based network behavior insight of cellular networks for {Industry} 4.0
	applications,'' \emph{IEEE Trans. Ind. Informat.}, vol.~16, no.~2, pp.
	1310--1320, 2020.
	
	\bibitem{Jiang3}
	D.~{Jiang}, L.~{Huo}, Z.~{Lv}, H.~{Song}, and W.~{Qin}, ``A joint
	multi-criteria utility-based network selection approach for
	vehicle-to-infrastructure networking,'' \emph{IEEE Trans. Intell. Transp.
		Syst.}, vol.~19, no.~10, pp. 3305--3319, 2018.
	
	\bibitem{Nair2004}
	G.~N. Nair and R.~J. Evans, ``Stabilizability of stochastic linear systems with
	finite feedback data rates,'' \emph{SIAM J. Control Optim.}, vol.~43, no.~2,
	pp. 413--436, Jul. 2004.
	
	\bibitem{LiuGC}
	W.~Liu, P.~Popovski, Y.~Li, and B.~Vucetic, ``Real-time wireless networked
	control systems with coding-free data transmission,'' in \emph{Proc. IEEE
		Globecom}, 2019.
	
	\bibitem{LiuJIoT}
	W.~{Liu}, P.~{Popovski}, Y.~{Li}, and B.~{Vucetic}, ``Wireless networked
	control systems with coding-free data transmission for {Industrial IoT},''
	\emph{IEEE Internet Things J.}, vol.~7, no.~3, pp. 1788--1801, Mar. 2020.
	
	\bibitem{TCP}
	L.~Schenato, B.~Sinopoli, M.~Franceschetti, K.~Poolla, and S.~S. Sastry,
	``Foundations of control and estimation over lossy networks,'' \emph{Proc.
		IEEE}, vol.~95, no.~1, pp. 163--187, Jan. 2007.
	
	\bibitem{KangTWC}
	K.~Huang, W.~Liu, M.~Shirvanimoghaddam, Y.~Li, and B.~Vucetic, ``Real-time
	remote estimation with hybrid {ARQ} in wireless networked control,''
	\emph{IEEE Trans. Wireless Commun.}, vol.~19, no.~5, pp. 3490--3504, 2020.
	
	\bibitem{KangJIoT}
	K.~{Huang}, W.~{Liu}, Y.~{Li}, B.~{Vucetic}, and A.~{Savkin}, ``Optimal
	downlink-uplink scheduling of wireless networked control for {Industrial
		IoT},'' \emph{IEEE Internet Things J.}, vol.~7, no.~3, pp. 1756--1772, Mar.
	2020.
	
	\bibitem{AlexLeong}
	A.~S. Leong, S.~Dey, and D.~E. Quevedo, ``Transmission scheduling for remote
	state estimation and control with an energy harvesting sensor,''
	\emph{Automatica}, vol.~91, pp. 54 -- 60, 2018.
	
	\bibitem{AlexLeong2}
	A.~S. Leong, A.~Ramaswamy, D.~E. Quevedo, H.~Karl, and L.~Shi, ``Deep
	reinforcement learning for wireless sensor scheduling in cyber–physical
	systems,'' \emph{Automatica}, vol. 113, pp. 1--8, 2020.
	
	\bibitem{Burak}
	B.~{Demirel}, A.~{Ramaswamy}, D.~E. {Quevedo}, and H.~{Karl}, ``Deepcas: A deep
	reinforcement learning algorithm for control-aware scheduling,'' \emph{IEEE
		Control Syst. Lett.}, vol.~2, no.~4, pp. 737--742, Apr. 2018.
	
	\bibitem{Eisen}
	M.~{Eisen}, M.~M. {Rashid}, K.~{Gatsis}, D.~{Cavalcanti}, N.~{Himayat}, and
	A.~{Ribeiro}, ``Control aware radio resource allocation in low latency
	wireless control systems,'' \emph{IEEE Internet Things J.}, vol.~6, no.~5,
	pp. 7878--7890, 2019.
	
	\bibitem{Eisen2}
	M.~{Eisen}, K.~{Gatsis}, G.~J. {Pappas}, and A.~{Ribeiro}, ``Learning in
	wireless control systems over nonstationary channels,'' \emph{IEEE Trans.
		Signal Process.}, vol.~67, no.~5, pp. 1123--1137, May 2019.
	
	\bibitem{Nair09}
	P.~Minero, M.~Franceschetti, S.~Dey, and G.~N. Nair, ``Data rate theorem for
	stabilization over time-varying feedback channels,'' \emph{IEEE Trans. Autom.
		Control}, vol.~54, no.~2, pp. 243--255, Feb. 2009.
	
	\bibitem{KangICC}
	K.~Huang, W.~Liu, Y.~Li, and B.~Vucetic, ``To retransmit or not: Real-time
	remote estimation in wireless networked control,'' in \emph{Proc. IEEE ICC},
	2019.
	
	\bibitem{tradeoffs}
	B.~Soret, P.~Mogensen, K.~I. Pedersen, and M.~C. Aguayo-Torres, ``Fundamental
	tradeoffs among reliability, latency and throughput in cellular networks,''
	in \emph{2014 IEEE Globecom Workshops (GC Wkshps)}, Dec. 2014, pp.
	1391--1396.
	
	\bibitem{Polyanskiy}
	Y.~Polyanskiy, H.~V. Poor, and S.~Verdu, ``Channel coding rate in the finite
	blocklength regime,'' \emph{IEEE Trans. Inf. Theory}, vol.~56, no.~5, pp.
	2307--2359, May 2010.
	
	\bibitem{DanielBernoulli}
	P.~K. Mishra, D.~Chatterjee, and D.~E. Quevedo, ``Stabilizing stochastic
	predictive control under bernoulli dropouts,'' \emph{IEEE Trans. Autom.
		Control}, vol.~63, no.~6, pp. 1579 -- 1590, Jun. 2018.
	
	\bibitem{DanielTradeoff}
	B.~Demirel, V.~Gupta, D.~E. Quevedo, and M.~Johansson, ``On the trade-off
	between communication and control cost in event-triggered dead-beat
	control,'' \emph{IEEE Trans. Autom. Control}, vol.~62, no.~6, pp. 2973--2980,
	Jun. 2017.
	
	\bibitem{KangWCL}
	K.~{Huang}, W.~{Liu}, Y.~{Li}, B.~{Vucetic}, and A.~{Savkin}, ``Wireless
	feedback control with variable packet length for industrial {IoT},''
	\emph{IEEE Wireless Commun. Lett.}, (early access) 2020.
	
	\bibitem{Polyfeedback}
	Y.~Polyanskiy, H.~V. Poor, and S.~Verd{\'u}, ``Feedback in the non-asymptotic
	regime,'' \emph{IEEE Trans. Inf. Theory}, vol.~57, no.~8, pp. 4903--4925,
	Aug. 2011.
	
	\bibitem{KangGC}
	K.~Huang, W.~Liu, Y.~Li, and B.~Vucetic, ``To sense or to control: Wireless
	networked control using a half-duplex controller for {IIoT},'' in \emph{Proc.
		IEEE Globecom}, 2019.
	
	\bibitem{liu2020remote}
	W.~Liu, D.~E. Quevedo, Y.~Li, K.~H. Johansson, and B.~Vucetic, ``Remote state
	estimation with smart sensors over markov fading channels,'' \emph{submitted
		to IEEE Trans. Autom. Control}, 2020.
	
	\bibitem{Pareto}
	J.~{Wang}, C.~{Jiang}, H.~{Zhang}, Y.~{Ren}, K.~{Chen}, and L.~{Hanzo},
	``Thirty years of machine learning: The road to {Pareto}-optimal wireless
	networks,'' \emph{IEEE Commun. Surveys Tuts.}, 2020.
	
	\bibitem{BookInfo}
	T.~Cover and J.~Thomas, \emph{Elements of Information Theory}.\hskip 1em plus
	0.5em minus 0.4em\relax John Wiley and Sons, 2006.
	
\end{thebibliography}

\end{document}